\documentclass[a4paper]{article}

\usepackage{amsmath}
\usepackage{amssymb}
\usepackage{amsthm}
\usepackage{amsfonts}

\usepackage{bbm}

\usepackage{url}

\usepackage{cleveref}

\usepackage{titling}
\DeclareMathOperator{\Aut}{Aut}

\DeclareMathOperator{\M}{M}

\newcommand{\comment}[1]{}

\newcommand{\card}[1]{|#1|}

\renewcommand{\o}{\omega}

\newcommand{\myset}[1]{\{1,\dots,#1\}}

\newcommand{\id}{\mathbbm{1}}

\newtheorem{proposition}{Proposition}

\newtheorem{lemma}{Lemma}

\theoremstyle{definition}
\newtheorem{definition}{Definition}
\newtheorem{example}{Example}

\usepackage[affil-it]{authblk}
\title{MacWilliams Extension Theorem for MDS additive codes}
\author{Serhii Dyshko
	\thanks{Electronic address: \texttt{dyshko@univ-tln.fr}}}
\affil{Institut de math\'ematiques de Toulon, Universit\'e de Toulon, France}

\begin{document}
\maketitle
\begin{abstract}
	The MacWilliams Extension Theorem states that each linear isometry of a linear code extends to a monomial map. Unlike the linear codes, in general, additive codes do not have the extension property. In this paper, an analogue of the extension theorem for additive codes in the case of additive MDS codes is proved. More precisely, it is shown that for almost all additive MDS codes their additive isometries extend to isometries of the ambient space.
\end{abstract}

\section{Introduction}
The MacWilliams Extension Theorem does not have a general analogue neither for nonlinear codes nor for additive codes. Nevertheless, in \cite{aug1} and \cite{kov1} the authors observed some classes of nonlinear codes for which an analogue of the extension theorem for nonlinear codes holds. In \cite{d1} we proved that the extension theorem for additive codes holds for the codes with the length not greater than some boundary value. There we also proved that, in general, this result cannot be improved.

In this paper, our main objective is to study the extendibility of additive isometries of MDS (maximum distance separable) additive codes. It appears that for almost all MDS codes, except the case of codes of dimension 2, the extension theorem holds, see \Cref{thm-mds-not2}. In the exceptional case, when code dimension equals 2, we can improve the general result of \cite{d1} and increase the bound on the code length, see \Cref{thm-mds-when2}.

Additionally, we observed an extension theorem for additive isometries of linear codes that are not linear isometries. The results are formulated in \Cref{thm-mwet-linear-codes}.

\section{Preliminaries}
Let $L$ be a finite field and let $n$ be a positive integer. Consider a Hamming space $L^n$. The MacWilliams Extension Theorem gives a full description of linear isometries of codes in $L^n$. It states that each linear isometry of a linear code in $L^n$ extends to a monomial map. A map $f: L^n \rightarrow L^n$ is called \emph{monomial} if it acts by permutation of coordinates and multiplications of coordinates by nonzero scalars. Note that monomial maps describe all isometries of the full Hamming space.

A general analogue of the MacWilliams Extension theorem does not exist for nonlinear codes. There exists an isometry of a nonlinear code that does not extend to an isometry of the whole space (see \cite{aug1}).

In \cite{d1} we observed a generalization of the MacWilliams Extension Theorem for the class of additive codes.
A code in $L^n$ is called \emph{additive} if it is an additive subgroup of $L^n$.
An \emph{additive isometry} of an additive code $C$ is an isometry that is a group homomorphism. Evidently, a map $f$ is an additive isometry if and only if $f$ preserves the Hamming weight.

\begin{example}\label{ex-1}
	Consider an additive code $C = \{ (0, 0, 0), (1,1,0), (\omega,0,1), (\omega^2,1,1) \}$ in $\mathbb{F}_4^3$, where $\mathbb{F}_4 = \{ 0, 1, \omega, \omega^2\}$ and $\omega + 1 = \omega^2$.
	Define a map $f: C \rightarrow \mathbb{F}_4^3$ in the following way:
	$f\big((0,0,0)\big) = (0,0,0)$, $f\big((1,1,0)\big) = (0,\omega^2,\omega)$, $f\big((\omega,0,1)\big) = (1,0,1)$ and $f\big((\omega^2,1,1)\big) = (1,\omega^2,\omega^2)$. The map $f$ is additive and it preserves the Hamming weight. Therefore $f$ is an additive isometry of the additive code $C$ in $\mathbb{F}_4^3$. Both codes $C$ and $f(C)$ are not $\mathbb{F}_4$-linear.
\end{example}
Let $K$ be a subfield of $L$. Along with the additive codes we will speak about \emph{$K$-linear codes}, i.e. codes that are $K$-linear subspaces of $L^n$. The notions of additive and $K$-linear codes in $L^n$ are in some sense equivalent.
Any $K$-linear code is additive and, conversely, any additive code is $\mathbb{F}_p$-linear, where $p$ is the characteristic of $L$. If $K = L$, a $K$-linear code is linear.
Obviously, any $K$-linear isometry is additive and any additive isometry is $\mathbb{F}_p$-linear. 

\begin{definition}\label{def-general-monomial}
A map $f: L^n \rightarrow L^n$ is called \emph{$K$-monomial} if there exist a permutation $\pi \in S_n$ and automorphisms $g_1, \dots, g_n \in \Aut_K(L)$ such that for all $u \in L^n$,
\begin{equation*}
f(u) = f\big((u_1, u_2, \dots, u_n)\big) = \big(g_1 (u_{\pi(1)}), g_2(u_{\pi(2)}), \dots, g_n( u_{\pi(n)})\big)\;.
\end{equation*}
\end{definition}
It is an easy exercise to prove that a map $f: L^n \rightarrow L^n$ is $K$-monomial if and only if it is a $K$-linear isometry.

An extension theorem for $K$-linear code isometries does not hold in general. For any pair of fields $K \subset L$ there exists a $K$-linear code and there exists a $K$-linear isometry of this code that cannot be extended to a $K$-monomial map. The example observed in \cite{d1} follows.

\begin{example}\label{example-general-unextendible}
	Consider two $K$-linear codes $C_1 = \langle v_1, v_2 \rangle_K$ and $C_2 = \langle u_1, u_2 \rangle_K$ of the length $\card{K} + 1$ with
	\begin{equation*}
	\left(
	\begin{matrix}
	v_1 \\
	v_2
	\end{matrix}\right) =
	\left(
	\begin{matrix}
	0 & 1& 1& \dots & 1\\
	1 & x_1 & x_2 & \dots & x_{\card{K}}
	\end{matrix}\right)
	\xrightarrow{f}
	\left(
	\begin{matrix}
	0 & 1& 1& \dots & 1\\
	0 & \o & \o & \dots & \o
	\end{matrix}\right) =
	\left(
	\begin{matrix}
	u_1 \\
	u_2
	\end{matrix}\right) \;,
	\end{equation*}
	where $x_i \in K$ are all different and $\o \in L \setminus K$.
	Define a $K$-linear map $f: C_1 \rightarrow C_2$ on the generators of $C_1$ in the following way: $f(v_1) = u_1$ and $f(v_2) = u_2$.
	The map $f$ is an isometry.
	But, there is no $K$-monomial transformation that acts on $C_1$ in the same way as the map $f$. The first coordinate of all vectors in $C_2$ is always zero, but there is no such all-zero coordinate in $C_1$.
\end{example}

However, we are able to prove an extension theorem for $K$-linear codes of short length. In \cite{d1} we proved the following.

\begin{proposition}\label{thm-small-length-mwet}
	Let $K \subset L$ be a pair of finite fields and let $n \leq \card{K}$. Any $K$-linear isometry of a $K$-linear code in $L^n$ extends to a $K$-monomial map.
\end{proposition}
\begin{proof}
	See \cite{d1}.
\end{proof}

According to \Cref{example-general-unextendible}, the result of \Cref{thm-small-length-mwet} cannot be improved in general. The aim of this paper is to improve this result for some classes of $K$-linear codes. Of particular interest are MDS additive codes. The description of the main technique that we use follows.

Denote the degree of the extension $[L:K] = m$. The finite field $L$ is a vector space over $K$. Fix a $K$-linear basis $b_1, \dots, b_m \in L$ of $L$ over $K$. For a positive integer $k$ and a vector-column $\vec{v} \in L^k$, let $\vec{v}_1, \dots, \vec{v}_m \in K^k$ be the expansion of $\vec{v}$ in the basis. This means that $\vec{v} = \sum_{i = 1}^m b_i \vec{v}_i$. Define a \emph{column space} $V \subseteq K^k$ of the vector $\vec{v}$ as the $K$-span of vectors, $V = \langle \vec{v}_1, \dots, \vec{v}_m \rangle_K$. Obviously, $0 \leq \dim_K V \leq m$.

\newcommand{\V}{\mathcal{V}}
\newcommand{\U}{\mathcal{U}}

Let $C$ be a $K$-linear code in $L^n$ and let $f: C \rightarrow L^n$ be a $K$-linear map.
Fix a $K$-linear basis $c_1, \dots, c_k \in L^n$ of $C$.
Let $A \in \M_{k \times n}(K)$ be a matrix with the rows $c_1,\dots, c_k$ and let $V_i \subseteq K^k$ denote the column space of the $i$th column of $A$, for $i \in \myset{n}$. 
Call $\V = (V_1, \dots, V_n)$ a \emph{tuple of spaces} of $C$. In \cite{d1} we proved an important formula for the dimension of a code, $\dim_K C = \dim_K \sum_{i=1}^n V_i$.

Let $f: C \rightarrow L^n$ be a $K$-linear map. Let $B \in \M_{k \times n}(K)$ be a matrix with $i$th row $f(c_i)$, $i \in \myset{k}$. Denote $\U = (U_1, \dots, U_n)$ the tuple of spaces of $f(C)$, where $U_i$ is the column space of the $i$th column of $B$, for $i \in \myset{n}$. Note that the $K$-linear span of the rows of $B$ equals to the code $f(C)$.

Call $(\U,\V)$ a \emph{pair of tuples} that corresponds to the code $C$ and the map $f$. We say that $\V$ and $\U$ are \emph{equivalent}, and denote $\U \sim \V$, if there exists a permutation $\pi \in S_n$, such that $V_i  = U_{\pi(i)}$, for all $i \in \myset{n}$.

Recall for the pair of sets $X \subseteq Y$ the indicator function $\id_X : Y \rightarrow \{0,1\}$ is defined as $\id_X(x) = 1$ for $x \in X$ and $\id_X(x) =0$ otherwise. In \cite{d1} we proved the following.
\begin{proposition}\label{thm-isometry-criterium}
Let $C \subseteq L^n$ be a $K$-linear code and let $f: C \rightarrow L^n$ be a $K$-linear isometry. Let $(\U,\V)$ be a pair of tuples that correspond to $C$ and $f$. The map $f$ is an isometry if and only if
\begin{equation}\label{eq-main-counting-space}
\sum_{i=1}^n \frac{1}{\card{V_i}} \id_{V_i}=
\sum_{i=1}^n \frac{1}{\card{U_i}} \id_{U_i}\;.
\end{equation}
The map $f$ extends to a $K$-monomial map if and only if $\U \sim \V$.
\end{proposition}
\begin{proof}
	See \cite{d1}.
\end{proof}

A solution $(\U,\V)$ of \cref{eq-main-counting-space} is called \emph{trivial} if $\U \sim \V$, and \emph{nontrivial} otherwise.
According to \Cref{thm-isometry-criterium}, a $K$-linear code isometry extends to a $K$-monomial map if and only if the corresponding solution $(\U,\V)$ is trivial. 

In the case of linear code isometries, when $K = L$, we can easily prove the MacWilliams Extension Theorem using \Cref{thm-isometry-criterium}. Indeed, by the construction, in the case $[L:K]=1$, the spaces that appear in \cref{eq-main-counting-space} are either lines or zero spaces. Hence, it is easy to see that a solution of \cref{eq-main-counting-space} can be only trivial.  

Let $K$ be a proper subfield of $L$. Previously, in \cite{d1}, we proved that there exists a nontrivial solution of \cref{eq-main-counting-space} if and only if $m \geq \card{K} +1$. \Cref{thm-small-length-mwet} immediately follows from this fact and \Cref{thm-isometry-criterium}.

In the following sections we use the following notation. Let $K \subset L$ be a pair of finite fields, let $n$ be a positive integer and let $C$ be a $K$-linear code in $L^n$. Denote $q = \card{K}$, $k = \dim_K C$ and $m = [L:K]$. Let $\V = (V_1,\dots,V_n)$ be a tuple of spaces of $C$. If there is considered a $K$-linear map $f: C \rightarrow L^n$, let $(\U,\V)$ be a pair of tuples that correspond to $C$ and $f$, where $\U = (U_1,\dots, U_n)$.

\section{Extendibility of additive isometries of MDS codes}
In coding theory there is a famous Singleton bound according to which the cardinality of a code $C$ in $L^n$ is not greater than $\card{L}^{n - d + 1}$, where $d$ is the minimum distance of the code. The code is called MDS if $\card{C} = \card{L}^{n - d + 1}$.

In this section we assume that $C$ is a $K$-linear MDS code of dimension $k$ over $K$. Since $q^k = \card{C} = \card{L}^{n- d + 1} = {(q^m)}^{n - d +1}$, obviously, $k = m (n - d +1)$. Denote $k_L = n - d +1$, so that $k = k_Lm$. Note that $k_L = \log_{\card{L}}\card{C}$ represents an analogue of the dimension of a code in linear case.

\begin{lemma}\label{lemma-mds-properties}
For each subset $I \subseteq \myset{n}$, $\dim_K \sum_{i \in I} V_i= m\min \{ k_L, \card{I} \}$.
\end{lemma}
\begin{proof}
It is a well-known fact that a code with minimal distance $d$ is MDS if and only if deleting any $d - 1$ column we get a new code of the same cardinality (see \cite[p.~319]{macwilliams}). Let $I \subseteq \myset{n}$ be a set with $k_L$ elements. Let $c_1, \dots, c_k \in L^n$ be a $K$-linear basis of $C$ that correspond to the tuple of spaces $\V$ of the code. Consider the $K$-linear basis $c_1', \dots, c_k'$ of a new code $C'$, where each basis vector $c_i'$ is formed from $c_i$ by puncturing the coordinates with indexes $\myset{n} \setminus I$, for $i \in \myset{k}$. The tuple of the spaces $\V' = (V_1', \dots, V_{k_L}')$ that corresponds to the basis $c_1', \dots, c_k'$ contains only spaces from $\V$ with indexes from $I$. Hence, using the formula for the dimension of a code, $\dim_K \sum_{i \in I} V_i = \dim_K \sum_{i = 1}^{k_L} V_i' = \dim_K C = k_Lm$. Moreover, since for all $i \in \myset{n}$, $\dim_K V_i \leq m$ and $k_Lm = \dim_K \sum_{i \in I} V_i \leq \sum_{i \in I} \dim_K V_i = \card{I} m = k_Lm$, we have $\dim_K V_i = m$, for all $i \in \myset{n}$. Evidently, if $\card{I} > k_L$, then $\dim_K\sum_{i \in I} V_i = k_Lm$.

Let $\card{I} < k_L$ and let $J\subseteq \myset{n}$ be a subset, such that $I \subset J$ and $\card{J} = k_L$. Assume that $\dim_K \sum_{i \in I} V_i < m \card{I}$.
Then $\sum_{i \in J} V_i = \sum_{i \in I} V_i + \sum_{i \in J \setminus I} V_i$ and $m k_L = \dim_K \sum_{i \in J} V_i \leq \dim_K \sum_{i \in I} V_i + \dim_K \sum_{i \in J \setminus I} V_i < m \card{I} + m (\card{J} - \card{I}) = m \card{J} = m k_L$. By contradiction, $\dim_K \sum_{i \in I} V_i \geq m \card{I}$. Since $\dim_K \sum_{i \in I} V_i \leq m \card{I}$, we get the statement of the proposition.
\end{proof}

\Cref{lemma-mds-properties} particularly states that $V_i \cap V_j = \{0\}$ for all $i \neq j$.

\begin{proposition}\label{thm-mds-not2}
Let $K \subset L$ be a pair of finite fields. Let $C$ be a $K$-linear MDS code in $L^n$ with $k_L \neq 2$. Any $K$-linear isometry of $C$ extends to a $K$-monomial map.
\end{proposition}
\begin{proof}
If $k_L = 1$ the statement of the proposition is obvious.

Let $k_L \geq 3$ and therefore $n \geq 3$.
Let $f: C \rightarrow L^n$ be a $K$-linear isometry.
Assume that $f$ does not extend to a $K$-monomial map. By \Cref{thm-isometry-criterium}, the pair $(\U,\V)$ is a nontrivial solution of \cref{eq-main-counting-space}. Using the same idea as in the proof of \Cref{thm-small-length-mwet} (see \cite{d1}), we can assume, after a proper reindexing, that there exists a nontrivial covering
$V_1 = \bigcup_{i=1}^t V_1 \cap U_i$, where $\{0\} \subset V_1 \cap U_i \subset V_1$, for all $i \in \myset{t}$, $t \geq q+1$. Note that $q = \card{K} \geq 2$ and therefore $t \geq 3$.
 
Let $a, b \in K^k$ be such that $a \in V_1 \cap U_1$, $b \in V_1 \cap U_2$ and $a,b \neq 0$. The map $f$ is an isometry, which implies the code $f(C)$ is MDS.
By \Cref{lemma-mds-properties}, $U_1 \cap U_2 = \{0\}$. We have $a \notin U_2$, $b \notin U_1$, $a+b \notin U_1$ and $a+b \notin U_2$. The element $a + b$ is nonzero since otherwise $a = -b \in U_2$. Also, $a+b \in V_1$ and $a+ b \in U_1+U_2$. In the covering $V_1\cap U_i$, $i \in \myset{t}$ , of $V_1$ there are at least $3$ nonzero spaces, so there exists a space, without loss of generality let it be $U_3$, such that $a+ b \in U_3$. Hence $U_3 \cap (U_1 + U_2) \neq \{0\}$.

From \Cref{lemma-mds-properties}, $\dim_K (U_1 + U_2 + U_3) = m\min\{k_L,3\} = 3m$ and thus $U_1 \cap (U_1 + U_2) = \{0\}$. By the contradiction, $f$ extends to a $K$-monomial map.
\end{proof}

For the case $k_L = 2$, the approach presented in \Cref{thm-mds-not2} fails. But we still can use the same idea to improve the result of \Cref{thm-small-length-mwet}.

Let $V$ be a vector space over $K$ of dimension $m$. \emph{Partition} of $V$ is a collection of proper subspaces of $V$, such that any nonzero vector from $V$ belongs to exactly one subspace from the collection. By $\sigma(m)$ we denote the minimal possible number of subspaces in the partition of $V$. In \cite{project-coverings} there are observed different properties of partitions and, particularly, the properties of the value $\sigma(m)$ for different $m$, and there is also mentioned the general lower bound (with the reference to the result of Beutelspacher \cite{beutelspacher}), $\sigma(m) \geq q^{\lceil \frac{m}{2} \rceil} + 1$.

\begin{proposition}\label{thm-mds-when2}
Let $K \subset L$ be a pair of finite fields. Let $C$ be a $K$-linear MDS code in $L^n$ with $k_L = 2$ and $n \leq q^{\lceil \frac{m}{2} \rceil}$. Each $K$-linear isometry of $C$ extends to a $K$-monomial map.
\end{proposition}
\begin{proof}
Assume that $f: C \rightarrow L^n$ is an unextendible $K$-linear isometry. Therefore the pair $(\U,\V)$ is a nontrivial solution of \cref{eq-main-counting-space}. 
As in the proof of \Cref{thm-mds-not2}, we can assume that the space $V_1$ is covered nontrivially, $V_1 = \bigcup_{i=1}^t V_1 \cap U_i$, where $n \geq t$. Since the code $f(C)$ is also MDS, by \Cref{lemma-mds-properties}, any two different spaces $U_i$ and $U_j$ intersect in zero. Therefore, $V_1 \cap U_i$, for $i \in \myset{t}$, is a partition of $V_1$ and $n \geq t \geq \sigma(m) > q^{\lceil \frac{m}{2} \rceil}$. By contradiction, the statement of the proposition holds.
\end{proof}

\section{Extendibility of additive isometries of linear codes}
Whereas the classical MacWilliams Extension theorem describes linear isometries of linear codes in $L^n$ it says nothing about the extendibility of nonlinear isometries of linear codes, particularly, it gives no information about additive isometries of a linear code.

Let the code $C$ be an $L$-linear code in $L^n$. In this section we study the extendibility of $K$-linear isometries of the code $C$, considered as a $K$-linear code. 

Denote by $k_L$ the dimension of $C$ over $L$. Let $A_L \in \M_{k_L \times n}(L)$ be a generator matrix of a linear code $C$. The rows of the matrix $A_L$ form an $L$-linear basis of $C$.
Let $b_1, \dots, b_m$ be a $K$-linear basis of $L$ over $K$. Denote by $b_i A_L$ the matrix formed from $A_L$ by multiplying each matrix entry by the scalar $b_i$, for all $i \in \myset{m}$. Consider a matrix $A = (b_1 A_L^T | \dots | b_m A_L^T)^T \in \M_{mk_L \times n}(L)$, that is a vertical concatenation of the matrices $b_i A_L$, $i \in \myset{m}$. It is easy to see that the rows of $A$ form a $K$-linear basis of $C$. Therefore $k = \dim_K C$ equals to $k_L m$. Recall $\V = (V_1, \dots, V_n)$ the tuple of spaces of $A$.

\begin{lemma}\label{lemma-linear-dims}
For each $i \in \myset{n}$, $\dim_K V_i = m$ or $\dim_K V_i = 0$. For all $i \neq j$, $V_i$ and $V_j$ either coincide or intersect in zero.
\end{lemma}
\begin{proof}
Prove the first part.
Let $i \in \myset{n}$. The dimension of the space $V_i$ equals to zero if and only if the corresponding $i$th coordinate in the code is zero. Assume that the $i$th column of the generator matrix $A_L$ contains a nonzero element $x \in L\setminus\{0\}$. Then the $i$th column of the generator matrix $A$ contains the elements $b_1 x, \dots, b_m x \in L \setminus \{0\}$, which form a $K$-linear basis of $L$. Hence the column space $V_i$ has dimension $\dim_K V_i \geq m$ and thus $\dim_K V_i = m$.

Prove the second part. Let $i \neq j \in \myset{n}$ be such that $i$th and $j$th columns of $A_L$ are nonzero. If $i$th vector-column can be obtained from $j$th vector column by multiplication by a nonzero scalar, then, because such a multiplication is a $K$-linear automorphism, $V_i = V_j$. Otherwise, the punctured code $C'$ obtained from $C$ by holding only $i$th and $j$th coordinates has dimension $2$ over $L$. Following the idea of \Cref{lemma-mds-properties}, two column spaces of $C'$ are $V_i$ and $V_j$. Considering $C'$ as a $K$-linear code, we have $2m = \dim_K C' = \dim_K (V_i + V_j)$ and therefore $V_i \cap V_j = \{0\}$.
\end{proof}

\begin{proposition}\label{thm-mwet-linear-codes}
Let $K \subset L$ be a pair of finite fields. Let $C$ be an $L$-linear code in $L^n$ with the length $n \leq q^{\lceil \frac{m}{2} \rceil}$ and let $f: C \rightarrow L^n$ be a $K$-linear isometry such that $f(C)$ is an $L$-linear code. The map $f$ extends to a $K$-monomial map.
\end{proposition}
\begin{proof}
The proof is almost the same as in \Cref{thm-mds-when2}. Assume that $f: C \rightarrow L^n$ is an unextendible $K$-linear isometry and thus the pair $(\U,\V)$ is a nontrivial solution of \cref{eq-main-counting-space}. 
We can assume that the space $V_1$ is covered nontrivially, $V_1 = \bigcup_{i=1}^t V_1 \cap U_i$, where $n \geq t$. From \Cref{lemma-linear-dims}, since $V_1$ is nonzero space, $\dim_K V_1 = m$. The code $f(C)$ is also $L$-linear, and thus from \Cref{lemma-linear-dims} the spaces $U_1, \dots, U_n$ or coincide or intersect in zero. Using the same arguments as in the proof of \Cref{thm-mds-when2}, $n > q^{\lceil \frac{m}{2} \rceil}$, and therefore, by contradiction, the statement of the proposition is true.
\end{proof}

\end{document}